\documentclass[journal,comsoc]{IEEEtran} %

\usepackage{amsmath}
\usepackage{amssymb}
\usepackage{amsthm}
\usepackage{array}
\usepackage{fixmath}
\usepackage{graphicx}
\usepackage{cite}
\usepackage{color}
\usepackage{changes}
\usepackage{algorithm}
\usepackage{algorithmic}
\usepackage{blindtext}
\usepackage{acronym}
\usepackage{soul}
\usepackage{hyperref}
\hypersetup{
    colorlinks=true,
    allcolors=blue,
    }

\DeclareMathAlphabet{\mathbit}{OML}{cmr}{bx}{it}

\newcommand{\B}[1]{\mathbit{#1}}
\newacro{MIMO}{Multiple-Input and Multiple-Output}
\newacro{MAC}{Multiple Access Channel}
\newacro{SNR}{Signal-to-Noise Ratio}
\newacro{CSI}{Channel State Information}
\newacro{MMSE}{Minimum Mean Squared Error}
\newacro{MSE}{mean squared error}
\newacro{pdf}{probability density function}
\newacro{SOC}{Second Order Cone}
\newacro{SINR}{Signal-to-Interference-Noise Ratio}
\newacro{BC}{Broadcast Channel}
\newacro{BS}{base station}
\newacro{MU}{mobile users}
\newacro{RF}{radio frequency}
\newacro{AWGN}{additive white Gaussian noise}
\newacro{MSO}{matrix shrinkage operator}
\newacro{mmWave}{millimeter-wave}
\newacro{ULA}{uniform linear array}
\newacro{LSR}{least square relaxation}
\newacro{LSV}{left singular vectors}
\newacro{MO}{manifold optimization}
\newacro{AoD}{Angle of Departure}
\newacro{AoA}{Angle of Arrival}
\newacro{PCS}{partially-connected structure }
\newacro{FCS}{fully-connected structure }
\newacro{ADMM}{alternating direction method of multipliers }
\newacro{OFDM}{Orthogonal Frequency Division Multiplexing}
\newacro{UAV}{unmanned aerial vehicle}
\newacro{LOS}{line of sight}
\newacro{SVD}{singular value decomposition}
\newacro{LISA}{linear successive allocation}
\newacro{ITA}{iterative trace-based algorithm}
\newacro{MRT}{maximum ratio transmission}
\newacro{RCCA}{Rank-Constrained Coordinate Ascent}
\newacro{ISI}{inter-symbol interference}
\newacro{ICI}{inter-carrier interference}

\DeclareMathOperator{\Hermitian}{H}

\newcommand{\He}{{\Hermitian}}
\DeclareMathOperator{\Exp}{\mathbb{E}}

\DeclareMathOperator{\trace}{tr}
\DeclareMathOperator{\rank}{rank}

\newcommand{\eul}{{\text{e}}}
\newcommand{\imj}{{\text{j}}}

\DeclareMathOperator{\diag}{diag}

\DeclareMathOperator*{\argmax}{argmax}
\DeclareMathOperator{\Frob}{F}

\def\B{\boldsymbol}

\usepackage[numbers,sort&compress]{natbib}

\newtheorem{Proposition}{Proposition}

\makeatletter
\def\blfootnote{\xdef\@thefnmark{}\@footnotetext}
\makeatother

\usepackage{tikz}
\usetikzlibrary{calc}

\newcommand{\positiontextbox}[4][]{%
  \begin{tikzpicture}[remember picture,overlay]
    \node[inner sep=3pt, fill=yellow,align=left,draw,line width=1pt,#1] at ($(current page.north west) + (#2,-#3)$) {\parbox{.95\paperwidth}{#4}};
  \end{tikzpicture}%
}

\begin{document}

\onecolumn
\begingroup

\setlength\parindent{0pt}
\fontsize{14}{14}\selectfont

\vspace{1cm} 
\textbf{This is an ACCEPTED VERSION of the following published document:}

\vspace{1cm} 
J. P. González-Coma, Ó. Fresnedo and L. Castedo, ``A Rank-Constrained Coordinate Ascent
Approach to Hybrid Precoding for the Downlink of Wideband Massive (MIMO) Systems'',
\textit{IEEE Transactions on Vehicular Technology}, vol. 72, no. 12, pp. 15953-15966, Dec. 2023,
doi: 10.1109/TVT.2023.3293933.

\vspace{1cm} 
Link to published version: \url{https://doi.org/10.1109/TVT.2023.3293933}

\vspace{3cm} 

\textbf{General rights:}

\vspace{1cm} 
\textcopyright 2023 IEEE. This version of the article has been accepted for publication, after peer review. \href{https://creativecommons.org/licenses/by-nc-nd/4.0/}{Personal use of this material is permitted. Permission from IEEE must be obtained for all other uses, in any current or future media, including reprinting/republishing this material for advertising or promotional purposes, creating new collective works, for resale or redistribution to servers or lists, or reuse of any copyrighted component of this work in other works.}
\twocolumn
\endgroup
\clearpage

\title{\LARGE{A rank-constrained coordinate ascent approach to hybrid precoding for the downlink of wideband massive \ac{MIMO} systems}} %

\author{José P. González-Coma, Óscar Fresnedo,~\IEEEmembership{Member,~IEEE,}
	 Luis Castedo,~\IEEEmembership{Senior Member,~IEEE}
  \thanks{J.P. Gonz\'alez-Coma is with Defense University Center at the Spanish Naval Academy, Plaza de Espa\~na s/n, 36920, Mar\'in, Spain. 
  	
  	O. Fresnedo and L. Castedo are with Dept. Computer Engineering \& CITIC Research Center, University of A Coru\~na, Campus de Elvi\~na s/n, 15071, A Coru\~na.}}
 
 \markboth{IEEE TRANSACTIONS ON VEHICULAR TECHNOLOGY}%
 {Shell \MakeLowercase{\textit{et al.}}: A Sample Article Using IEEEtran.cls for IEEE Journals}

\maketitle

\positiontextbox{10.75cm}{27cm}{\footnotesize \textcopyright 2023 IEEE. This version of the article has been accepted for publication, after peer review. Personal use of this material is permitted. Permission from IEEE must be obtained for all other uses, in any current or future media, including reprinting/republishing this material for advertising or promotional purposes, creating new collective works, for resale or redistribution to servers or lists, or reuse of any copyrighted component of this work in other works. Published version:
\url{https://doi.org/10.1109/TVT.2023.3293933}}

\begin{abstract}
An innovative approach to hybrid analog-digital precoding for the downlink of wideband massive \ac{MIMO} systems is developed. The proposed solution, termed  \ac{RCCA}, starts seeking the full-digital precoder that maximizes the achievable sum-rate over all the frequency subcarriers while constraining the rank of the overall transmit covariance matrix. The frequency-flat constraint on the analog part of the hybrid precoder and the non-convex nature of the rank constraint are circumvented by transforming the original problem into a more suitable one, where a convenient structure for the transmit covariance matrix is imposed. Such structure makes the resulting full-digital precoder particularly adequate for its posterior analog-digital factorization. An additional problem formulation to determine an appropriate power allocation policy according to the rank constraint is also provided. The numerical results show that the proposed method outperforms baseline solutions even for practical scenarios with high spatial diversity.

\end{abstract}
\IEEEpeerreviewmaketitle
\begin{IEEEkeywords}
massive MIMO, hybrid precoding, wideband, rank constraints.
\end{IEEEkeywords}

\section{Introduction}

\IEEEPARstart{M}{assive} \ac{MIMO} is a fundamental technology to reach the foreseen challenging requirements of data throughput, spectral efficiency and user accommodation of the future wireless communication systems  \cite{bjornson2017massive,Lu2014}. In this context, a large range of communication scenarios can benefit from using large arrays of antennas, including mobile networks \cite{Larsson14_magazine,Busari18}, vehicular-to-infrastructure (V2I) communications \cite{Pfadler20,Zhang16},  control of  unnamed aerial vehicles (UAVs) \cite{Bjornson_uavs,Chandhar_18}, or machine-to-machine (M2M) communications \cite{Lee21}.  However, unlike conventional \ac{MIMO}, having one dedicated \ac{RF} chain per antenna element is unaffordable in massive \ac{MIMO} because of its extremely high cost and power consumption \cite{Larsson14_magazine}. Consequently, massive \ac{MIMO} transceivers usually have a hybrid digital-analog structure which combine an analog \ac{RF} precoding network, with a limited number of \ac{RF} chains, and a digital baseband processing unit \cite{liang2014low,Sohrabi2016_narrowband,Gao2016}.
In this work, we consider hybrid precoding for the downlink of a wideband massive \ac{MIMO} system. The design of high performance hybrid precoders is significantly more challenging in wideband situations because analog \ac{RF} precoders are frequency-flat, i.e., their frequency response is constant over all the available bandwidth \cite{sohrabi2017hybrid}. 

In the literature, wideband hybrid precoding has been addressed for both single-user \cite{Alkhateeb2016_wideband,Yu2016,Park2017,Chen21} and multiuser \cite{sohrabi2017hybrid,Kong15, gonzalez2019hybrid, Yuan19, GaoCheLiu21} massive \ac{MIMO} systems. For instance, the authors in \cite{sohrabi2017hybrid} exploit the similarity between the channel responses at different subcarriers for \ac{mmWave} systems with a large number of antennas at both ends. Thus, the performance of the solutions in \cite{sohrabi2017hybrid} strictly relies on the assumption of  high similarity among the structure of channel matrices at different subcarriers.
In \cite{Kong15}, hybrid wideband precoders are designed using an alternating optimization algorithm based on the equivalence between sum-rate optimization  and weighted sum \ac{MSE} minimization. This solution, however, is limited to single-stream transmissions and single-antenna users. The authors in \cite{gonzalez2019hybrid} extend the use of  \ac{LISA} to the design of wideband hybrid precoders by exploiting the common subspace structure for the subcarriers of each user. A new stream is allocated to a given user at each iteration of the procedure, which causes less interference to the remaining users. 
However, this approach degrades severely for channels with high frequency selectivity, as it quickly consumes all the available spatial degrees of freedom. In \cite{Yuan19}, the authors propose an alternating algorithm to minimize the Frobenius norm between an unconstrained full-digital solution and the wideband hybrid precoders, but this approach only applies under the assumptions of single-stream transmissions and equal power allocation on each subcarrier. The work in \cite{GaoCheLiu21} proposes a design of the hybrid precoders for  wideband \ac{mmWave} massive \ac{MIMO} systems based on maximizing the mutual information with  a block diagonalization procedure. In this case, the frequency-dependent baseband precoders must be able to completely cancel the multiuser interference, which is a difficult task in general because of the low dimensions of these precoding matrices for a reasonable number of \ac{RF} chains. This effect also presents in other methods pursing the same of strategy of interference removal in the baseband stage, e.g. \cite{Alkhateeb15,sohrabi2017hybrid,ChChJiHa19}. Two-timescale hybrid precoding approaches are considered in \cite{Liu18,Lui_conference2018,Mai18}, where the frequency-flat analog precoding network is designed according to the channel statistics, whereas the digital baseband matrices are updated according to the time-varying \ac{CSI} for each channel realization. The performance of these strategies for wideband scenarios relies on the assumption that the channel statistics are similar over the whole bandwidth, and hence the resulting \ac{RF} precoder is adequate for all the subcarriers. This assumption, however, is unrealistic in large bandwidth scenarios.

On the other hand, most of the previous work assume \ac{LOS} channel models with a limited number of reflection paths \cite{sohrabi2017hybrid,gonzalez2019hybrid,venugopal2017}, which makes the inter-user interference easier to handle. These assumptions, however, are not practical in \ac{mmWave} indoor and sub-6 GHz scenarios, where the channel matrices are spatially rich. This effect hinders the similarity of the channel spatial features over different subcarriers, and it is even more pronounced when considering large bandwidths  \cite{Wang2019}. In these scenarios, the design of hybrid precoders with a frequency-flat analog component is a challenging problem.

In this paper, we address the general case of hybrid precoding for the downlink of a wideband massive \ac{MIMO} system. In particular, we develop an innovative strategy based on limiting the rank of the full-digital precoders while exploiting the spatial features of the wideband channel matrices. This decision is motivated by the lack of accuracy of the methods that approximate the digital solutions for multiuser wideband scenarios, which suffer from a drastic dimensional reduction that severely limits the achievable performance. Accordingly, we design digital precoders with a rank restriction which effectively reduces the number of allocated data streams. The aim is to obtain digital precoders more suitable for the subsequent hybrid factorization. In this sense, since the subspace spanned by the rank-constrained digital design has a limited dimension, the corresponding hybrid design will be  capable of handling wideband multiuser channels more effectively. Few existing works have considered the incorporation of rank constraints to address the limited number of available \ac{RF} chains \cite{BoLeHaVa16,GoFrCa21}. Indeed, the authors in \cite{BoLeHaVa16} showed a relationship between the rank of the full-digital precoders and the accuracy of its corresponding hybrid factorization.
Building on this idea, if the rank of a full-digital solution is adjusted according to the number of available \ac{RF} chains, its corresponding analog and digital counterparts would achieve similar performance results.
 
 However, the incorporation of  rank constraints into the problem formulation for wideband scenarios significantly complicates the design of hybrid precoders due to the flat-frequency nature of the analog \ac{RF} part. To circumvent this issue, we propose an original approach termed as Rank-Constrained Coordinate Ascent (RCCA), %
 where all the user transmit covariance matrices are alternatively updated while constraining the rank of the overall wideband transmit covariance matrix. This is accomplished thanks to the use of a common frequency-flat structure for each user transmit covariance matrix, combined with a power allocation that is jointly designed for all the users and subcarriers. While the common transmit covariance matrix structure accounts for the inter-user interference at all the subcarriers, the power allocation is intended to select the most promising data streams of each user.

The main contributions of this work are summarized as follows:
\begin{itemize}
	\item An \ac{RCCA} algorithm for the design of hybrid precoders for the downlink of wideband massive \ac{MIMO} systems is developed. Unlike existing approaches, this algorithm does not employ a codebook or minimizes the distance to optimized full-digital solutions according to some metric. Instead, it determines full-digital solutions with a particular structure, which is suitable for their posterior analog-digital factorization. The proposed design also aims at exploiting the potential similarity among the user subcarriers. However, unlike other previous approaches  \cite{sohrabi2017hybrid,venugopal2017,GaoCheLiu21}, the proposed solution addresses a general channel model, since it does not make explicit assumptions on the frequency response of the subcarriers (e.g., \ac{mmWave} scenarios can assume that the channel response at different subcarriers is rather similar). %
	
	\item The initial rank-constrained optimization problem is transformed into a sequence of steps which provide a useful insight about the desired structure of the user transmit covariance matrices. Also, the transmit power distribution among users and subcarriers is specifically designed to incorporate the rank constraint, and to consider the inter-user interference across the entire channel bandwidth.
	
	\item Under the assumption of a rank constraint design of the hybrid precoders, we prove that the allocation of exactly the same number of streams as that of available \ac{RF} chains leads to the optimal system performance in terms of achievable sum-rate.

\end{itemize}

\section{System Model}\label{Sec:model}
Let us consider the downlink of a wideband massive \ac{MIMO} system where a \ac{BS} with $M$ antennas communicates to $U$ users with $R$ antennas each. The number of RF chains at the \ac{BS} is limited to $N_\text{RF} \ll M$. As in  \cite{liang2014low,Sohrabi2016_narrowband,Gao2016}, we also assume that $R\ll M$. \ac{OFDM} modulation with $K$ subcarriers is utilized for transmission over the wideband channels.

Let $\B{s}_u[k]$ be the data symbols transmitted to user $u$ over subcarrier $k$, with $u = 1, \ldots, U$ and $k=1, \ldots, K$. We assume that the \ac{BS} allocates $N_{s,u}$ data streams to user $u$. Hence, $\B{s}_u[k] \in \mathbb{C}^{N_{s,u} \times 1}$ and the total number of streams transmitted by the \ac{BS} is $N_s = \sum_{u=1}^U N_{s,u}$. We also assume that $\Exp[\B{s}_u[k]\B{s}_j^\He[k]]=\mathbf{0}$ for $j\neq u$, and $\Exp[\B{s}_u[k]\B{s}_u^\He[l]]=\mathbf{0}$ for $k\neq l$. %
The user data is linearly precoded using a limited number of RF chains at the BS, following a hybrid digital-analog architecture, with $\B{P}^u_\text{BB}[k]\in\mathbb{C}^{N_\text{RF}\times N_{s,u}}$ being the digital baseband precoder for user $u$ at the $k$-th subcarrier. Hence, the discrete-time digitally precoded OFDM  symbol vector $\B{s}_u^{n}\in\mathbb{C}^{N_\text{RF}}$ reads as 
	\begin{equation} \label{eq:dataSymbol}
		\B{s}_u^{n} = \sum_{k=1}^{K}\B{P}^u_\text{BB}[k]\B{s}_u[k]\eul^{\imj \frac{2\pi(k-1) n}{K}},~~~n=0,1, \ldots, K-1.
	\end{equation} 
	Note the number of OFDM symbols in the time domain is also $K$. We consider that data symbols $\B{s}_u^{n}$ in \eqref{eq:dataSymbol} modulate with the square-root raised cosine pulse-shaping filter $ p_{\text{srrc}}(t)$ such that $p_{\text{srrc}}(t)*p_{\text{srrc}}(-t)=p_{\text{rc}}(t)$ is the raised cosine pulse signal which  fulfills the Nyquist zero \ac{ISI} criterion. Next, the  resulting continuous-time vector signal inputs the analog precoding stage modeled by the matrix $\B{P}_\text{RF}\in\mathcal{P}^{M \times N_\text{RF}}$. This matrix is frequency-flat and common to all users and $\mathcal{P}$ represents the set of feasible RF precoder matrices with unit-modulus entries. Correspondingly, the time-domain signal vector for user $u$ is given by 
	\begin{equation}
		\B{s}_u(t)=\sum_{n=0}^{N_\text{c}-1}\sum_{k=1}^{K}\B{P}_\text{RF}\B{P}^u_\text{BB}[k]\B{s}_u[k] p_{\text{srrc}}(t-nT_s)\eul^{\imj \frac{2\pi(k-1) n}{K}},  %
		\label{eq:Bs(t)}
	\end{equation}
	where $T_s$ is the sampling period. The linear combination of the signals $\B{s}_u(t)$ for the different users is transmitted through the wireless channel, for which we consider the following general geometrical model for the channel responses \cite{Sayeed2002,VeGoHe19,Chen21}
	\begin{equation}
		\B{H}_u(t)[k]=\gamma\sum_{p=1}^{N_p}\alpha_{u,p}\delta(t-\tau_{u,p})\B{a}_{\text{R},u}(\theta_{u,p})[k]\B{a}_{\text{T}}^\He(\phi_{u,p})[k],\label{eq:channelModelTime}
	\end{equation}
	with  $\gamma=\sqrt{MR/N_p}$ being the power normalization factor, $N_p$ the number of channel paths, $\alpha_{u,p}$ the complex gain, $\tau_{u,p}$ the delay, and  $\B{a}_{\text{R},u}(\theta_{u,p})[k]\in\mathbb{C}^{R\times1}$ and $\B{a}_{\text{T}}(\phi_{u,p})[k]\in\mathbb{C}^{M\times1}$ the frequency dependent array response vectors of the user and the BS at the corresponding angular directions $\theta_{u,p}$ and $\phi_{u,p}$ for subcarrier $k$, respectively. Therefore, the discrete-time equivalent channel response at the $d$-delay tap, which includes the effects of the transmit and receive pulse-shaping filters, can be written as 
	\begin{equation}
		\B{H}_u^d[k]=\gamma\sum_{p=1}^{N_p}\alpha_{u,p}p_{\text{rc}}(dT_s-\tau_{u,p})\B{a}_{\text{R},u}(\theta_{u,p})[k]\B{a}_{\text{T}}^\He(\phi_{u,p})[k],\label{eq:channelModelFreq}
	\end{equation}
	where $d\in\{1,\ldots , D\}$, being $D$ the maximum number of delay taps. From the former equation, the spatial relationship among the different subcarriers is apparent, specially for those with similar indices $k$. Moreover, the usage of \acp{ULA} at both ends is
considered here, although the proposed scheme is applicable to other configurations. Therefore, the array response vectors are given by \cite{gonzalez2019hybrid,ChChJiHa19}
\begin{equation}
\B{a}_{\text{T}}(\phi)[k]
=\frac{1}{{\sqrt{M}}}\Big[
1,e^{j\frac{2\pi c}{f[k]}\;d \; \text{sin}\;\phi},
\ldots,e^{j\frac{2\pi c}{f[k]}\;d  \left(  {M}-1\right)\;\text{sin}\;\phi} \Big]^{T},
\label{eq:steeringVector}   
\end{equation}
where $f[k]$ is the frequency at the $k$-th subcarrier, $c$ is the speed of light, and $d$ is the inter-antenna spacing. The carrier frequency  for the $k$-th sub-band is given by $f[k]=f_c+\xi[k]$, where $f_c$ is the central frequency and $\xi[\ell]$ represents an offset with respect to $f_c$. If we assume a signal bandwidth of $B$, the frequency offset for the $k$-th subcarrier reads as $\xi[k]=(k-1-\frac{K-1}{2})\frac{B}{K}$. Note that the expression for $\B{a}_{\text{R}}(\theta)[k]$ is similar to that of $\B{a}_{\text{T}}(\phi)[k]$, and thus omitted for brevity. When the signal bandwidth satisfies $B\ll f_c$, $f[k]\approx f_c$ and the response vectors are not affected by the frequency offsets. On the contrary, when $B$ is comparable to $f_c$ the beam-squint effect arises in \eqref{eq:steeringVector}.

In order to suppress \ac{ISI} and \ac{ICI},  a cyclic prefix large enough is added to the \ac{OFDM} symbol in \eqref{eq:dataSymbol}. Moreover, the basis functions employed to build the time-domain signal in \eqref{eq:Bs(t)} are also chosen to constitute an orthogonal set on the subcarrier index $k$. Combining these two facts, the discrete-time equivalent of the signal received by user $u$ over subcarrier $k$, $\B{x}_u[k] \in \mathbb{C}^{R \times 1}$, can be expressed as \cite{rodriguez2007comunicaciones}
\begin{align}
\B{x}_u[k]&=\sum_{d=0}^{D-1}\B{H}_u^d[k]\sum_{i=1}^{U}\B{P}_\text{RF}\B{P}^i_\text{BB}[k]\B{s}_i[k]\eul^{-\imj\frac{2\pi (k-1)d}{K}}+\B{n}_u\notag\\
&=\B{H}_u[k]\sum_{i=1}^{U}\B{P}_\text{RF}\B{P}^i_\text{BB}[k]\B{s}_i[k]+\B{n}_u,
\label{eq:receivedSignal}
\end{align}
where $\B{H}_u[k]\in\mathbb{C}^{R\times M}$ is the $u$-th user channel frequency response matrix at subcarrier $k$, $\B{H}_u[k]=\sum_{d=0}^{D-1}\B{H}_u^d[k]\eul^{-\imj\frac{2\pi (k-1)d}{K}}$, and $\B{n}_u\sim\mathcal{N}_\mathbb{C}(\mathbf{0},\sigma_n^2\mathbf{I}_{R})$ is the receiving \ac{AWGN}.  The available transmit power at the \ac{BS} is $P_{\text{tx}}$, and hence the system \ac{SNR} is defined as SNR = $P_{\text{tx}}/\sigma_n^2$.  Without loss of generality, we assume $\sigma_n^2=1$ for simplicity.

\section{Initial Problem Formulation}
Considering the communication model described in the previous section, we are interested in exploring the design of  hybrid precoders which maximize the theoretically achievable sum-rate. This  metric is extensively used in the literature to measure the overall system performance \cite{Alkhateeb15,sohrabi2017hybrid,Gao2016,Elbir20,Yuan19}.
Assuming Gaussian signaling and superposition coding with successive interference cancellation, the achievable rate for user $u$ over subcarrier $k$ is given by
\begin{align*}
R_u[k] = &\log_2\det\left(
\mathbf{I}_R%
+\B{X}_u^{-1}[k]\B{H}_u[k]\B{P}^u_\text{H}[k]\B{P}^{u,\He}_\text{H}[k]\B{H}_u^{\He}[k]\right),
\end{align*}
with $\B{X}_u[k]=\sum_{i> u}\B{H}_u[k]\B{P}^i_\text{H}[k]\B{P}^{i,\He}_\text{H}[k]\B{H}_u^{\He}[k]+\mathbf{I}_{R}$ containing the interference plus noise terms for a given user decoding order, and with $\B{P}^i_\text{H}[k]=\B{P}_\text{RF}\B{P}^i_\text{BB}[k]$ the hybrid precoding matrix for the $i$-th user at subcarrier $k$. Note that the expression above represents the individual rates which could be obtained theoretically with a configuration based on successive interference cancellation, and therefore they would actually lead to an upper bound  for the achievable sum-rate in practical systems. This assumption also matches the common strategy in which the precoder is first designed assuming an ideal receiver, and then the corresponding receiver is calculated for this precoder, thus enabling decoupled  precoding and decoding problems \cite{LeeWaChYu15,Sohrabi2016_narrowband,sohrabi2017hybrid,AyRaAbPiHe14,YuShZhLe16,ChChJiHa19}.

As already mentioned, we aim at determining the precoder matrices $\B{P}_\text{RF}$ and $\B{P}_\text{BB}^u[k]$, $u=1, \ldots, U$, $k=1, \ldots, K$, which maximize the achievable sum-rate given by $$R=\frac{1}{K}\sum_{k=1}^K\sum_{u=1}^UR_u[k],$$  for a fixed transmit power $P_\text{tx}$. Hence, the corresponding constrained optimization problem can be formulated as
\begin{align}
\max_{\B{P}_\text{RF},\B{P}_\text{BB}^u[k]}& R\quad\text{s.t.}\quad\sum_{k=1}^K\sum_{u=1}^U\|\B{P}_\text{RF}\B{P}_\text{BB}^u[k]\|_{\Frob}^2\leq P_\text{tx},\; \B{P}_\text{RF}\in\mathcal{P}. %
\label{eq:problFormOriginal}
\end{align}
This difficult optimization problem can be addressed by first determining the full-digital precoders which maximize the achievable sum-rate $R$, and next factorizing them into their \ac{RF} and baseband counterparts. Thus, by defining the overall hybrid transmit covariance matrix as
\begin{align}
	\B{Q}_\text{H}=&\B{P}_\text{RF}\left[\B{P}^1_\text{BB}[1]\B{P}^{1,\He}_\text{BB}[1],\ldots,\B{P}^1_\text{BB}[K]\B{P}^{1,\He}_\text{BB}[K],\right.\notag\\
	&\left.\ldots,\B{P}^U_\text{BB}[K]\B{P}^{U,\He}_\text{BB}[K]\right]\left(\mathbf{I}_{UK}\otimes\B{P}^\He_\text{RF}\right),
	\label{eq:covHybrid}
\end{align} 
we can observe that a side effect of using hybrid precoding is the imposition of the restriction $\rank(\B{\B{Q}_\text{H}})\leq N_\text{RF}$. Building on this idea and contrary to approaches like, e.g.  \cite{YuShZhLe16,Yuan19}, we impose a rank constraint on the design of the full-digital precoders with the aim of obtaining a more accurate factorization in a later step.   
For that, we first express the achievable sum-rate in the downlink as follows
\begin{align*}
R^\text{D}=\frac{1}{K}\sum_{u=1}^{U}\sum_{k=1}^{K}\log_2\det\left(
\mathbf{I}_R +\B{Y}_u^{-1}[k]\B{H}_u[k]\B{Q}_u[k]\B{H}_u^{\He}[k]\right),
\end{align*}
where the matrix $\B{Q}_u[k]$ is introduced to represent the full-digital transmit covariance matrix  for the $k$-th subcarrier of user $u$, and  $\B{Y}_u[k]=\sum_{i> u}\B{H}_u[k]\B{Q}_i[k]\B{H}_u^{\He}[k]+\mathbf{I}_{R}$. Note that the super-index $^D$ refers to the downlink. Next, we introduce the following  rank-constrained formulation
\begin{align}
\max_{\{\B{Q}_u[k]\}_{u=1,k=1}^{U,K}} R^\text{D} \quad\text{s.t.}\quad&\sum_{k=1}^K\sum_{u=1}^U\trace(\B{Q}_u[k])\leq P_\text{tx}, \notag \\
&\rank(\B{Q})\leq N_\text{RF},\label{eq:rankFormulation}
\end{align}
where  $\B{Q}=[\B{Q}_1[1],\ldots,\B{Q}_1[K],\ldots,\B{Q}_U[K]]\in\mathbb{C}^{M\times MUK}$ is the overall transmit covariance matrix. In particular, the rank constraint in \eqref{eq:rankFormulation} forces the precoder design to satisfy $\rank(\B{Q})\leq N_\text{RF}$; that is, the columns of $\B{Q}$ will span a subspace of dimension smaller than or equal to $N_\text{RF}$ on the vector space $\mathbb{C}^M$. Observe that the maximum subspace dimension achievable by the span of $\B{Q}_\text{H}$ in \eqref{eq:covHybrid} also satisfies this restriction. Therefore, our strategy limits the performance of the full-digital solution to ensure that it can be accurately approximated by means of a factorization procedure. Since we are imposing a rank constraint but not restricting the entries of $\B{Q}$ in any way, the problem formulation in \eqref{eq:rankFormulation} can be interpreted as a relaxed version of the original one in \eqref{eq:problFormOriginal}, and hence the resulting sum-rate will actually be an upper bound to that obtained in \eqref{eq:problFormOriginal}. Note also that, contrary to other approaches like e.g. \cite{BoLeHaVa16,GoFrCa21}, we do not focus on selecting the best subset of users since all of them will be served if $K\leq N_\text{RF}$ holds.

\section{Reformulation of the optimization problem}
Solving the optimization problem \eqref{eq:rankFormulation} is challenging due to the coupling between the covariance matrices for different users and subcarriers and the presence of the rank restriction. These constraints are non-convex, and hence it is difficult to deal with them. For this reason, it is common to transform rank-constrained problems into other ones that are easier to handle by means of relaxed formulations or convex surrogates \cite{YuLa11,MaGoLi11}. In this case, we start leveraging the duality between the achievable sum-rate for the downlink and a virtual uplink \cite{ViJiGo03}. Assuming successive decoding, the achievable sum-rate for the dual virtual uplink is given by
\begin{align}
R^\text{U}&=\sum_{k=1}^K\log_2\det\left(\mathbf{I}_M+\sum_{u=1}^U\B{H}_u^\He[k]\B{S}_u[k]\B{H}_u[k]\right),\label{eq:sumrateMAC}
\end{align}
where $\B{H}_u^\He[k]$ and $\B{S}_u[k]$ are the uplink channel response and the uplink transmit covariance matrices for user $u$ over the $k$-th subcarrier, respectively. Observe that no particular user decoding order is assumed in \eqref{eq:sumrateMAC}. Therefore, different orderings differ on the individual achievable rates while leading to the same achievable sum-rate. Moreover, the result in \cite{ViJiGo03} guarantees that the uplink sum-rate $R^\text{U}$ is achieved in the downlink by selecting the transmit covariance matrices as $\B{Q}_u[k]=\B{\Delta}_u[k]\B{S}_u[k]\B{\Delta}_u^\He[k]$, where 
\begin{equation}
\B{\Delta}_u[k]=\B{B}_u^{-1/2}[k]\B{F}_u[k]\B{G}_u^\He[k]\B{A}_u^{1/2}[k]
\label{eq:UL-DLconv}
\end{equation} 
are the uplink-downlink conversion matrices. These matrices are determined via the auxiliary matrices
\begin{equation}
\begin{aligned}[b]
\B{A}_u[k] & =\mathbf{I}_R+\sum\nolimits_{i=1}^{u-1}\B{H}_u[k]\B{Q}_i[k]\B{H}_u^\He[k]\\
\B{B}_u[k] & =\mathbf{I}_M+\sum\nolimits_{i=u+1}^{U}\B{H}_i^\He[k]\B{S}_i[k]\B{H}_i[k]
\end{aligned}
\label{eq:ABduality}
\end{equation}
where  $\B{F}_u[k]\in\mathbb{C}^{M\times R}$ and $\B{G}_u[k]\in\mathbb{C}^{R\times R}$ result from the \ac{SVD} $\B{F}_u[k]\B{\Xi}_u[k]\B{G}_u^\He[k]=\B{B}_u^{-1/2}[k]\B{H}_u^\He[k]\B{A}_u^{-1/2}[k]$. Notice that $\B{A}_u[k]$ and $\B{B}_u[k]$ are obtained in sequential order, i.e., for $u=1,2,\ldots,U$, and independently for each subcarrier. 

By means of this duality,  the optimization problem in \eqref{eq:rankFormulation} is reformulated as 
\begin{align}
\max_{\{\B{S}_u[k]\succeq 0\}_{u=1,k=1}^{U,K}} R^\text{U}\quad\text{s.t.}\quad&\sum_{k=1}^K\sum_{u=1}^U\trace(\B{S}_u[k])\leq P_\text{tx}, \notag\\
&\rank(\B{W})\leq N_\text{RF},\label{eq:rankFormulationUL}
\end{align}
with $\B{W}=[\B{\Delta}_1[1]\B{S}_1[1]\B{\Delta}_1^\He[1],\ldots,\B{\Delta}_U[K]\B{S}_U[K]\B{\Delta}_U^\He[K]]$. Recall that the result in \cite{ViJiGo03} guarantees that the achievable sum-rate $R^\text{U}$ obtained with a set of uplink transmit covariance matrices $\B{S}_u[k]$ $\forall u,k$ is also achievable in the downlink using $\B{Q}_u[k]=\B{\Delta}_u[k]\B{S}_u[k]\B{\Delta}_u^\He[k]$ $\forall u,k$, i.e., $R^\text{U}=R^\text{D}$ under the same total power constraint $\sum_{k=1}^K\sum_{u=1}^U\trace(\B{S}_u[k])=\sum_{k=1}^K\sum_{u=1}^U\trace(\B{Q}_u[k])$. Moreover, \eqref{eq:rankFormulationUL} includes the uplink-downlink conversion matrices, $\B{\Delta}_u[k]$ $\forall k,u$, in the restriction $\rank(\B{W})\leq N_\text{RF}$ to ensure that the downlink rank restriction holds. Therefore, $\rank(\B{W})\leq N_\text{RF}$ in \eqref{eq:rankFormulationUL} is equivalent to $\rank(\B{Q})\leq N_\text{RF}$ in \eqref{eq:rankFormulation}. By virtue of the total power and the rank conditions, optimal solutions of \eqref{eq:rankFormulationUL} also lead to optimal solutions for the problem in \eqref{eq:rankFormulation}. Accordingly, we first solve \eqref{eq:rankFormulationUL} in the dual virtual uplink and next obtain the corresponding downlink precoders with the help of the uplink-downlink conversion matrices. Towards this aim, we develop two intermediate steps to circumvent the rank restriction in \eqref{eq:rankFormulationUL}. In the first step, a common spatial structure for the covariance matrices is proposed for all the frequencies of each user $u$ and, based on that decision, the second step determines the power allocation for each subcarrier $k$.

\subsection{Frequency-flat transmission covariance matrix structure}
\label{sec:flatStructure}
To address the first step, we must consider the involved structure of $\B{W}$ in \eqref{eq:rankFormulationUL}. A design of $\B{S}_u[k]$ independent for each subcarrier would definitely lead to large values for the rank of $\B{W}$. Moreover, in a more practical sense, a design of $\B{S}_u[k]$ exploiting all the available degrees of freedom is hardly achievable with a frequency-flat analog precoder design. Therefore, we will start by providing some insight regarding the structure of a frequency-flat transmit covariance matrix for a particular user. 

For that, let us first rewrite \eqref{eq:sumrateMAC} as $R^\text{U}=\sum_{u=1}^U R^\text{U}_u$. Moreover, we assume that user $u$ is decoded first, without loss of optimality, to obtain
\begin{equation}
R^\text{U}_u=\sum_{k=1}^K\log_2\det\left(\mathbf{I}_M+\B{H}_{\bar{u}}^\He[k]\B{S}_u[k]\B{H}_{\bar{u}}[k]\right),
\label{eq:perUserRateA}
\end{equation}
where $\B{H}_{\bar{u}}[k]=\B{H}_u[k](\sum_{i\neq u}\B{H}_i^\He[k]\B{S}_i[k]\B{H}_i[k]+\mathbf{I}_M)^{-\frac{1}{2}}$ represents the $u$-th user effective channel for subcarrier $k$. With the purpose of studying the desirable properties of a frequency-flat covariance matrix, we substitute $\B{S}_u[k]$ by its frequency-flat version $\B{S}_u$ in \eqref{eq:perUserRateA}, leading to
\begin{equation}
\bar{R}^\text{U}_{u}=\sum_{k=1}^K\log_2\det\left(\mathbf{I}_M+\B{H}_{\bar{u}}^\He[k]\B{S}_u\B{H}_{\bar{u}}[k]\right).
\label{eq:perUserRate}
\end{equation}
Fixing the transmit covariance matrices of the remaining users, i.e. $\B{S}_i[k]\;\forall k$ with $i\neq u$, the matrix $\B{S}_u$ which maximizes \eqref{eq:perUserRate} can be determined independently, thus ensuring an increase on the achievable sum-rate given by $\bar{R}^\text{U} = \sum_{i\neq u}^UR^\text{U}_{i}+\bar{R}^\text{U}_{u}$. 

Under the previous considerations, we can approach the problem in \eqref{eq:rankFormulationUL} piecemeal by stating  the next optimization problem for the update of the $u$-th covariance matrix, 
\begin{align}
\max_{\B{S}_u\succeq 0} \bar{R}^\text{U}_u\quad~\text{s.t.}\quad&\sum_{k=1}^K\sum_{i\neq  u}\trace(\B{S}_i[k])+\trace(\B{S}_u)\leq P_\text{tx},\label{eq:individualFormulationUL}
\end{align}
where the rank constraint is for the time being disregarded. The formulation in \eqref{eq:individualFormulationUL} leads to the following Lagrangian function
\begin{align}
L(\B{S}_u,\lambda_u,\B{\Lambda}_u)&=\bar{R}^\text{U}_u+\trace(\B{\Lambda}_u\B{S}_u)\\
&-\lambda_u\left(\sum_{k=1}^K\sum_{i\neq  u}\trace(\B{S}_i[k])+\trace(\B{S}_u)-P_\text{tx}\right),\notag
\end{align}
with $\lambda_u\geq 0$ and $\B{\Lambda}_u\succeq\mathbf{0}$ an auxiliary matrix to guarantee the constraints on $\B{S}_u$. In this scenario, we can readily obtain the Karush-Kuhn-Tucker (KKT) conditions as
\begin{align}
	\frac{\partial L(\B{S}_u,\lambda_u,\B{\Lambda}_u)}{\partial \B{S}_u}&=\B{\Lambda}_u-\lambda_u\mathbf{I}_R\label{eq:KKT}\\
	+\sum_{k=1}^K&\left(\mathbf{I}_R+\B{H}_{\bar{u}}[k]\B{H}_{\bar{u}}^\He[k]\B{S}_u\right)^{-1}\B{H}_{\bar{u}}[k]\B{H}_{\bar{u}}^\He[k]\nonumber,
\end{align}
with $\B{S}_u\B{\Lambda}_u=\B{\Lambda}_u\B{S}_u=\mathbf{0}$, and 
\begin{equation}
	\lambda_u\left(\sum_{k=1}^K\sum_{i\neq  u}\trace(\B{S}_i[k])+\trace(\B{S}_u)-P_\text{tx}\right)=0\notag.
\end{equation}
Now, multiplying the expression in \eqref{eq:KKT} times $\B{S}_u$ we get
\begin{equation}
\hspace*{-0.1cm}
\left(\sum_{k=1}^K\left(\mathbf{I}_R+\B{H}_{\bar{u}}[k]\B{H}_{\bar{u}}^\He[k]\B{S}_u\right)^{-1}\B{H}_{\bar{u}}[k]\B{H}_{\bar{u}}^\He[k]-\lambda_u\mathbf{I}_R\right)\B{S}_u=\mathbf{0}.%
\label{eq:KKTcondition}
\end{equation} 

\begin{Proposition}\label{th1}
	The KKT condition in \eqref{eq:KKTcondition} holds for $K=1$ and $\B{S}_u$ fulfilling  $\left(\mathbf{I}_R+\B{H}_{\bar{u}}[k]\B{H}_{\bar{u}}^\He[k]\B{S}_u\right)^{-1}\B{H}_{\bar{u}}[k]\B{H}_{\bar{u}}^\He[k]\succ\mathbf{0}$.
\end{Proposition}
\begin{proof}\label{proof-th1}
	See Appendix \ref{app:covariance_structure}. 
\end{proof}
It is apparent that finding a frequency-flat matrix $\B{S}_u$ fulfilling Prop. \ref{th1} simultaneously for all the subcarriers is difficult when $K>1$.  Notwithstanding, if these conditions approximately hold for each individual subcarrier, i.e., we select a matrix $\B{S}_u$ which approximates $\left(\mathbf{I}_R+\B{H}_{\bar{u}}[k]\B{H}_{\bar{u}}^\He[k]\B{S}_u\right)^{-1}\B{H}_{\bar{u}}[k]\B{H}_{\bar{u}}^\He[k]$ to a positive-definite matrix for $k=1,\ldots,K$,  the sum  $\sum_{k=1}^K\left(\mathbf{I}_R+\B{H}_{\bar{u}}[k]\B{H}_{\bar{u}}^\He[k]\B{S}_u\right)^{-1}\B{H}_{\bar{u}}[k]\B{H}_{\bar{u}}^\He[k]$ will also approach a positive-definite matrix. For that reason, we propose to select a covariance matrix $\B{S}_u$ satisfying
\begin{equation}
\sum_{k=1}^K\left(\mathbf{I}_R+\B{H}_{\bar{u}}[k]\B{H}_{\bar{u}}^\He[k]\B{S}_u\right)^{-1}\B{H}_{\bar{u}}[k]\B{H}_{\bar{u}}^\He[k]\approx \B{Z}_u\B{D}_u\B{Z}_u^\He,%
\label{eq:diagonal}
\end{equation} 
where $\B{D}_u$ is a diagonal matrix and $\B{Z}_u$ is a unitary matrix. In this case, the KKT condition in \eqref{eq:KKTcondition} will approximately be met.

Applying this procedure iteratively to update the covariance matrix for each of the users will increase the performance metric, as we discuss in Section \ref{sec:Alg}. Moreover since $\B{S}_u,\forall u$  are frequency-flat matrices, if these transmit covariance matrices are also feasible solutions of \eqref{eq:rankFormulationUL}, that they actually provide a lower bound for the achievable performance in \eqref{eq:rankFormulationUL}. 

\subsection{Frequency-selective power allocation}
\label{eq:power_alloc}

As explained in Prop. \ref{th1}, it is advisable to rewrite the $u$-th user frequency-flat covariance matrix as $\B{S}_u=\B{U}_u\B{\Psi}_u\B{U}_u^\He$, with $\B{U}_u$ unitary and $\B{\Psi}_u$ diagonal,  to deal with the requirement of a common analog precoder for all the subcarriers. In addition, utilizing a common basis $\B{U}_u$ for the $K$ subcarriers facilitates achieving a low rank condition for $\B{W}$ in \eqref{eq:rankFormulationUL}.  
However, taking into account the flexibility provided by the frequency-selective baseband precoders, it is reasonable to independently adjust the diagonal matrix $\B{\Psi}_u$ for each subcarrier. That is, we propose to design the frequency-selective covariance matrices $\B{S}_u[k]$ as  
\begin{equation}
\B{S}_u[k]=\B{U}_u\B{P}_u[k]\B{U}_u^\He,
\label{eq:S}
\end{equation}
with $\B{P}_u[k]\succeq\mathbf{0}$ a diagonal matrix with the power allocation for user $u$ and subcarrier $k$. The motivation for the design strategy given by \eqref{eq:S} is threefold:
\begin{enumerate}
	\item The similarity among the row subspaces for channel matrices corresponding to neighboring subcarriers is leveraged. Accordingly, $\B{U}_u$ is a basis used to span a vector space and $\B{P}_u[k]$ selects some of the columns to define an appropriate subspace for each subcarrier $k$. Thus, if channel subspace similarity is strong, \eqref{eq:diagonal} approximately holds individually for most of the subcarriers. 
	
	\item The conversion matrices in \eqref{eq:UL-DLconv} depend on the structure of the channel matrices (cf. \eqref{eq:ABduality}). As a consequence, the similarity among the channel matrix subspaces, together with the use of the common basis $\B{U}_u$, enables us to obtain a low rank $\B{W}$ in \eqref{eq:rankFormulationUL} even when $K$ is large, i.e., 
	\begin{equation}
	\rank(\B{W})\approx	\rank([\B{\Delta}_1[K_c]\B{U}_1,\ldots,\B{\Delta}_U[K_c]\B{U}_U]),
	\label{eq:rankApprox}
	\end{equation}
	where the approximation comes from considering $\B{\Delta}_u[k]\approx\B{\Delta}_u[K_c]$, $\forall k, u$, where $K_c$ is the central subcarrier. This subcarrier is the one presenting the minimum distance, in terms of frequency offset, with respect to the most distant subcarriers. Therefore, we can expect that the spatial features of $\B{\Delta}_u[k]$ with $k\neq K_c$ are similar to those of $\B{\Delta}_u[K_c]$. 
	
	\item Recall from \eqref{eq:problFormOriginal} that the transmit power in a hybrid architecture is distributed via the baseband precoders $\B{P}^u_\text{BB}[k]$. Hence, we can use the frequency-selective matrices $\B{P}_u[k]$ in \eqref{eq:S} to adjust the rank of the covariance matrices $\B{S}_u[k]$ by selecting the more promising per-user streams in a per-subcarrier basis.
\end{enumerate}

\vspace*{0.1cm}
According to the spatial characteristics described in Sec. \ref{sec:flatStructure}, we will use the covariance matrices in \eqref{eq:S} to address the optimization problem in \eqref{eq:rankFormulationUL} by considering the relaxed restriction in \eqref{eq:rankApprox}. Thus, we will address the problem by assuming that the rank constraint is imposed by means of the right term in \eqref{eq:rankApprox}.
In this context, it is still necessary to determine the power allocation matrices $\B{P}_u[k]$. It is worth remarking that the power allocation policy must be designed to effectively ensure that
\begin{align*}
	\rank(&[\B{\Delta}_1[K_c]\B{S}_1[1]\B{\Delta}_1^H[K_c],
	\ldots,\\&\B{\Delta}_U[K_c]\B{S}_U[K]\B{\Delta}_U^H[K_c]) \leq N_{\text{RF}},
\end{align*}
while searching for maximizing the sum-rate in \eqref{eq:rankFormulationUL}. In the following, we propose a strategy for obtaining $\B{P}_u[k]$  when $\B{U}_u$ is already given for $u=1,\ldots,U$.   

Let us start introducing the spatially truncated channels $\B{H}_{\bar{u}}{(\beta)}[k]=\left(\mathbf{I}_R-\sum_{n=1}^N\B{u}_{u,i_n}\B{u}_{u,i_n}^\He\right)\B{H}_{\bar{u}}[k]$, where $\B{u}_{u,c}$ denotes the $c$-th column of $\B{U}_u$. Recall that $ \B{H}_{\bar{u}}[k]$ represents the $u$-th user effective channel for subcarrier $k$. To define  the column index set $\mathcal{I}=\{i_1, \ldots, i_N\}$, consider the entries of the diagonal matrix
\begin{equation}
	\B{\Sigma}_u[k]=\diag(\B{U}_u^{\He}\B{H}_{\bar{u}}[k]\B{H}_{\bar{u}}^{\He}[k]\B{U}_u)
	\label{eq:channelGains}
\end{equation} 
as the potential equivalent channel gains for user $u$ at subcarrier $k$.  Hence, we build $\mathcal{I}$ by establishing priorities for the columns of the basis $\B{U}_u$ as
	\begin{equation}
	\begin{cases}
	i\in\mathcal{I}, & \text{if } \prod_{k=1}^K\big[\B{\Sigma}_u[k]\big]_{i,i}\geq\beta,\\
	i\notin\mathcal{I}, & \text{otherwise}.
	\end{cases}
	\end{equation}
That is, the column $i$ of  the matrix $\B{U}_u$ is selected when the product of the potential gains across all the subcarriers is larger than some threshold $\beta$. Otherwise, the stream associated to such a column is discarded for the whole frequency band. Accordingly, the value of $\beta$ is closely related to the rank expression in \eqref{eq:rankApprox}.  Building on this idea, we propose to find the power allocation matrices $\B{P}_u[k]$ $\forall u,k$ for given $\B{U}_u$ $\forall u$ by solving the optimization problem 
\begin{align}
&\hspace*{-.3cm} \max_{\{\B{P}_u[k]\}_{u,k=1}^{U,K}} \sum_{k=1}^K\sum_{u=1}^U\log_2\det\left(\mathbf{I}_M+(\B{V}_u{(\beta)}[k])^\He\B{P}_u[k]\B{V}_u{(\beta)}[k]\right)\;\notag\\
 &\text{s.t.} ~~\B{P}_u[k]\in\mathcal{D}^+ ,\forall u,k ~~\text{and} ~~\sum_{k=1}^K\sum_{u=1}^U\trace(\B{P}_u[k])\leq P_\text{tx},
\label{eq:ModifiedAlg}
\end{align}
where $\B{V}_u{(\beta)}[k]=\B{U}_u^\He\B{H}_{\bar{u}}{(\beta)}[k]$ and $\mathcal{D}^+$ is the set of diagonal matrices with non-negative entries. Note that the problem in \eqref{eq:ModifiedAlg} uses the spatial truncated channels $\B{H}_{\bar{u}}{(\beta)}[k]$ to address the rank constraint, which leads to a significant reduction in the number of equivalent channel gains considered for the power allocation procedure. Accordingly, a value of $\beta$ high enough entails power allocation matrices $\B{P}_u[k]$ such that the relaxed rank constraint in \eqref{eq:rankApprox} is satisfied. However, since the optimal value of $\beta$ is unknown beforehand, its determination will be included in our algorithmic solution.
Notice also that the covariance matrices obtained from  \eqref{eq:ModifiedAlg} are subject to more restrictive constraints than those in \eqref{eq:rankFormulationUL}, due to the imposed frequency-flat basis $\B{U}_u$. 

In summary, to tackle the optimization problem in \eqref{eq:rankFormulationUL}, we impose the particular structure for the user covariance matrices $\B{S}_u[k]$ in \eqref{eq:S}, according to the KKT conditions for the frequency-flat basis $\B{U}_u$. Next, we propose the optimization problem \eqref{eq:ModifiedAlg} to determine the frequency-selective $\B{P}_u[k]$ power allocation matrices.  In the ensuing section, we describe the proposed algorithm to implement this procedure.

\section{Proposed algorithm}
\label{sec:Alg}
In this section we develop an algorithmic approach to solve the problem in \eqref{eq:rankFormulationUL}. At each iteration $\ell$, the structure of the covariance matrices $\B{S}_u^{(\ell)}[k]$ is updated considering the matrices for the remaining users fixed, according to the strategy explained in Sec. \ref{sec:flatStructure}. To that end, the frequency-flat basis $\B{U}_u^{(\ell)}$ is determined as the \ac{LSV} of the matrix $[\B{H}_{\bar{u}}^{(\ell)}[1], \B{H}_{\bar{u}}^{(\ell)}[2], \ldots,  \B{H}_{\bar{u}}^{(\ell)}[K]]$, with the aim of capturing the common structure for the effective channels corresponding to user $u$ across the whole bandwidth. Note that this approach aims at satisfying the condition in \eqref{eq:diagonal}. Indeed, when the row subspaces for the equivalent channels are similar, the condition approximately holds. Moreover, assuming similar row subspaces for neighboring subcarriers is reasonable in the proposed scenario \cite{GaDaWaCh15,VeGoHe19,venugopal2017}. In the following, we will omit the super-index $(\ell)$ for notation simplicity. 

Next, in order to determine the power distribution among the users and subcarriers, i.e., $\B{P}_u[k]$ in \eqref{eq:S} for $u = 1, \ldots, U$  and $k = 1, \ldots, K$, we compute the equivalent channel gains  $\B{\Sigma}_u[k]$ in \eqref{eq:channelGains}. Note also that we propose to determine the power allocation according to \eqref{eq:ModifiedAlg}. Therefore, it is necessary to perform a search on $\beta$ and solve \eqref{eq:ModifiedAlg} to obtain a power allocation that maximizes the objective function and fulfills the rank constraint on $\B{W}$. To avoid this time-consuming search, we propose an alternative strategy. For that, it is important to remark that the objective function of \eqref{eq:ModifiedAlg} is upper bounded as follows
\begin{align}
		\sum_{u=1}^U\sum_{k=1}^K\log_2\det\left(\mathbf{I}_M+(\B{V}_u{(\beta)}[k])^\He\B{P}_u[k]\B{V}_u{(\beta)}[k]\right)\notag\\
		\leq\sum_{u=1}^U\log_2\prod_{k=1}^{K}\prod_{i=1}^{R}\left(1+\big[\B{P}_{u}[k]\big]_{i,i}\big[\B{\Sigma}_u[k]\big]_{i,i}\right),\label{eq:powerProblmUpper}
\end{align}
where $\big[\B{P}_{u}[k]\big]_{i,i}$ represents the $i$-th entry of the power allocation matrix $\B{P}_u[k]$.
Hence, we can use an iterative procedure where the largest equivalent gain, considering all the subcarriers jointly, is selected at each iteration. That is, $(j,i)=\max_{u,r}[\B{\Upsilon}_u]_{r,r}$, with $\B{\Upsilon}_u=\prod_{k=1}^K\B{\Sigma}_u[k], ~u=1, \ldots, U$. In this way, the tuples of indices $(j,i)$ corresponding to the largest wideband channel gains are sequentially added to the set $\mathcal{S}$, such that the index $j$ refers to the selected user and the index $i$ to the corresponding  stream. %
Therefore, the $i$-th column of $\B{U}_j$ is considered as an adequate candidate to allocate a certain amount of power in the subsequent design of $\B{P}_u[k]$. Furthermore, in order to account for the actual rank restriction in the downlink design, the  rank constraint is evaluated over the columns of $\B{\Delta}_j[K_c][\B{U}_j]$, that is, over the dual downlink covariance matrix given by the right term in \eqref{eq:rankApprox}. As mentioned, we will only consider the conversion matrix for the central subcarrier $\B{\Delta}_j[K_c]$. 

The procedure stops when $N_\text{RF}$ linearly independent vectors are selected, i.e., when $|\mathcal{S}|\geq N_\text{RF}$, and hence
\begin{align}
\hspace*{-0.1cm}\rank([\B{\Delta}_{j_1}[K_c]\big[\B{U}_{j_1}\big]_{:,i_1},\ldots,\B{\Delta}_{j_{|\mathcal{S}|}}[K_c]\big[\B{U}_{j_{|\mathcal{S}|}}]_{:,i_{|\mathcal{S}|}}])=N_\text{RF},
\nonumber
\end{align}
with the set $\mathcal{S} = \{(j_1, i_1), \ldots, (j_{|S|}, i_{|S|})\}$.
This assignment strategy allows us to obtain solutions for a relaxed version of the original formulation in \eqref{eq:rankFormulationUL} where the  the rank constraint is the right term  in \eqref{eq:rankApprox}. With the proposed scheme, we seek to minimize the loss of the posterior factorization procedure to obtain the hybrid precoder components.

Once the set $\mathcal{S}$ is determined, we calculate the power allocation vector $\B{p}\geq \mathbf{0}$ by using waterfilling over the equivalent channel gains $\big[\B{\Sigma}_j[k]\big]_{i,i}$ in \eqref{eq:channelGains}, just for the candidate tuples $(j,i)\in\mathcal{S}$. Using $\B{p}$, we compute the frequency-selective power allocation matrices $\B{P}_u[k]\;\forall u,k$, allowing us a flexible power distribution among the users and subcarriers. Indeed, the matrices $\B{P}_u[k]$ would switch off some of the users and/or subcarriers if their equivalent channel gains are poor.

An important remark is  that, when the KKT condition in \eqref{eq:KKTcondition} is satisfied, the upper bound of  \eqref{eq:powerProblmUpper} is tight. In that case, the power allocation procedure is hence equivalent to performing the search on $\beta$ and solve \eqref{eq:ModifiedAlg} for that value.  
\begin{algorithm}[!t]
	\caption{Rank-Constrained Coordinate Ascent (RCCA)}
	\label{alg:iterWfillRank}
	\small
	\begin{algorithmic}[1]
		\STATE Initialize: $\ell \leftarrow 0$,   $\{\B{S}_u^{(0)}[k]\}_{u=1,k=1}^{U,K}$, $K_c \leftarrow \lfloor K/2 \rfloor$, 
		\REPEAT 
		\FORALL{$u\in\{1,\ldots,U\}$}
		\STATE $\B{H}_{\bar{u}}^{(\ell)}[k]\leftarrow$ Update with $\B{H}_{\bar{u}}^{(\ell)}[k]$ from \eqref{eq:perUserRate}, $\forall k$%
		\STATE $\B{\Delta}_u^{(\ell)}[k]\leftarrow$ Compute using  \eqref{eq:UL-DLconv}, $\forall k$
		\STATE $\B{U}_u^{(\ell)}\leftarrow$ LSV of $[\B{H}_{\bar{u}}^{(\ell)}[1], \B{H}_{\bar{u}}^{(\ell)}[2], \ldots,  \B{H}_{\bar{u}}^{(\ell)}[K]]$ 
		\STATE $\B{\Sigma}_u^{(\ell)}[k]\leftarrow \diag(\B{U}_u^{(\ell),\He}\B{H}_{\bar{u}}^{(\ell)}[k]\B{H}_{\bar{u}}^{(\ell),\He}[k]\B{U}_u^{(\ell)})$, $\forall k$
		\STATE $\B{\Upsilon}_u^{(\ell)}\leftarrow \prod_{k=1}^K\B{\Sigma}_u^{(\ell)}[k]$
		\ENDFOR
		\STATE $r^{(\ell)}\leftarrow 0$, $\B{T}^{(\ell)}\leftarrow []$, $\mathcal{S}^{(\ell)}\leftarrow\emptyset$
		\WHILE {$r^{(\ell)}< N_\text{RF}$}
		\STATE $(j,i)\leftarrow\argmax_{j,i}\; [\B{\Upsilon}_j^{(\ell)}]_{i,i}$
		\STATE $\B{T}^{(\ell)}\leftarrow [\B{T}^{(\ell)} ~~ \B{\Delta}_j^{(\ell)}[K_c][\B{U}_j^{(\ell)}]_{:,i}]$
		\STATE $r^{(\ell)}\leftarrow \rank(\B{T}^{(\ell)})$
		\STATE $\mathcal{S}^{(\ell)}\leftarrow\mathcal{S}^{(\ell)}\cup\{(j,i)\}$  
		\STATE $[\B{\Upsilon}_j^{(\ell)}]_{i,i}\leftarrow 0$
		\ENDWHILE
		\STATE $\B{p}^{(\ell)}\leftarrow$ waterf. over $[\B{\Sigma}_j^{(\ell)}[k]]_{i,i},\forall k$ with $(j,i)\in\mathcal{S}^{(\ell)}$%
		\STATE $\{\B{S}_u^{(\ell+1)}[k]\}_{u=1}^U\leftarrow$ Compute using \eqref{eq:S} with $\B{U}_u^{(\ell)}, \B{p}^{(\ell)}$  
		\STATE $\ell \leftarrow \ell+1$
		\UNTIL{Convergence criterion is met}
		\STATE $\{\B{\Theta}_u[k]\}_{u=1,k=1}^{U,K} \gets$ Prec. matrices for $\{\B{S}_u^{(\ell)}[k]\}_{u=1,k=1}^{U,K}$
		\STATE $\B{P}_\text{RF}$, $\{\B{P}_\text{BB}^u[k]\}_{u=1,k=1}^{U,K}\gets$  \cite[Alg. 1]{gonzalez2018channel} with $\{\B{\Theta}_u[k]\}_{u=1,k=1}^{U,K}$
	\end{algorithmic}
\end{algorithm}

The proposed method, which we will term \ac{RCCA}, is a variation of a block coordinate ascent algorithm where the structure of the covariance matrices are independently determined by considering the remaining ones as fixed, and computing the power allocation jointly for all of them. This coupling in the power restriction, together with the absence of convex and closed properties for the feasible set imposed by the rank restriction, makes impossible to ensure monotonically increasing iterations \cite{Be99}. Nevertheless, it is possible to employ numerical techniques to ensure the convergence to a local optimum, like those employed in \cite{JiReViGo05}. Still, numerical results show good convergence behavior for this insightful algorithm.     
	
The proposed algorithm implicitly selects the number of data streams transmitted to each user at each subcarrier according to the power allocation procedure, considering the largest equivalent channel gains and satisfying the rank constraint with equality, according to the ensuing proposition
\begin{Proposition}\label{th2}
	The optimum overall transmit covariance matrix $\B{Q}$ which maximizes the achievable sum-rate in (6) satisfies $\rank(\B{Q})= N_\text{RF}$.
\end{Proposition}
\begin{proof}
	See Appendix \ref{app:streams}.
\end{proof}

To determine the overall computational complexity of Alg. \ref{alg:iterWfillRank}, we analyze the main steps of the procedure. The computational cost of obtaining the uplink-downlink conversion matrices in \eqref{eq:UL-DLconv} (steps 4 and 5), is bounded by $\mathcal{O}\left(M^2R\right)$ due to the computation of the matrix products, while the \ac{SVD} in step 6 and the product in step 7 present complexity orders of $\mathcal{O}\left(KMR^2\right)$ and $\mathcal{O}\left(MR^2\right)$, respectively. Next, finding the number of linearly independent columns in $\B{T}^{(\ell)}$ (step 14) is computationally cheap  with order $\mathcal{O}\left(MN_\text{RF}\right)$, whereas determining the uplink covariance matrices in step 19 presents a cost of $\mathcal{O}\left(R^3\right)$. Finally, obtaining the downlink precoding matrices of step 22 requires again computing matrix products, with a cost of about $\mathcal{O}\left(M^2R\right)$, and $\mathcal{O}\left(MN^2_\text{RF}\right)$ operations are necessary for the factorization procedure of step 23. If we consider only the dominant computational components, we obtain an overall computational complexity of $\mathcal{O}\left(\operatorname{max}\{KMR^2,M^2R\}\right)$.

\section{Simulation results}

In this section, the results of computer experiments carried out to illustrate the performance of the proposed \ac{RCCA} algorithm are presented. We also make comparisons with four benchmark approaches, namely: 1) hybrid \ac{LISA} \cite{gonzalez2019hybrid}, 2) the alternating minimization (AM)-based algorithm proposed in \cite{Yuan19}, 3) the hybrid precoding design of \cite{sohrabi2017hybrid}, and 4) a wideband extension of the \ac{ITA} proposed in \cite{GoFrCa21} for narrowband scenarios. The first one was chosen because it provides superior performance to traditional approaches for wideband hybrid precoding. The second and third solutions represent low-complexity alternatives for wideband hybrid precoding. Finally, the latter is a rank-constrained approach that iteratively solves a series of convex problems (cf. \cite{GoFrCa21}). In its wideband extension,  \ac{ITA} selects the best $N_{\text{RF}}$-dimensional subspace from the common space spanned by all the subcarriers, although the working principles significantly differ from the \ac{RCCA} algorithm. In addition, we assume successive interference cancellation at reception for all the approaches. 

\begin{table}[t]
	\centering
	\caption{{Simulation parameter settings.}}\label{tab:Sim1}
	\setlength{\tabcolsep}{5pt}
	\def\arraystretch{1.5}
	\begin{tabular}{|l|l|}
		\hline		
		{\textbf{Parameter}} & {\textbf{Value}}  				\\ \hline\hline
		Number of users & 		$U=6$ 	\\ \hline
		Central subcarrier frequency & 		$f_c=5$ GHz 		\\ \hline
		Bandwidth & 		$B=200$ MHz 		\\ \hline
		Number of \ac{BS} antennas & 			$M=32$ 				\\ \hline
		Distance between antennas & $d=\frac{\lambda}{2}$ m 	\\ \hline
		Number of antennas at the users & 	$R=2$  						\\ \hline
		Delay taps & 				 	$D=8$							\\ \hline
		Number of subcarriers &		$ K=64$ \\ \hline
		Number of RF chains &	 $N_{\text{RF}}=4$  \\ \hline
	\end{tabular}
\end{table}

\begin{figure}[t!]
	\centering
	\includegraphics[width=0.99\linewidth]{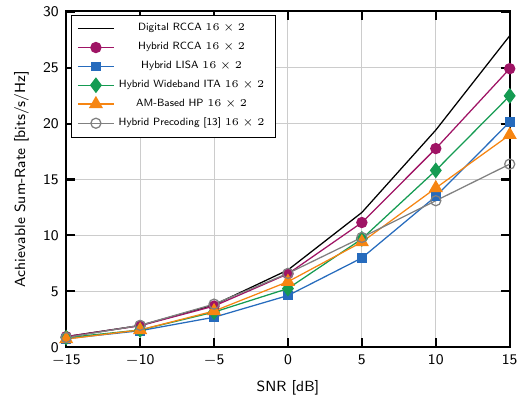}
	\caption{Achievable sum-rates of the different algorithms for  a downlink scenario with the following configuration: $M=16, R=2, U=6, N_{\text{RF}}=4, N_p = 16$ and $K=16$.}
	\label{fig:sumRates_differentMethods}
\end{figure}

Unless explicitly stated, the main simulation parameters are those in Table \ref{tab:Sim1}. The results were averaged over $200$ channel realizations generated according to the geometrical channel model described in \eqref{eq:channelModelFreq}. For all the experiments, we initialize the covariance matrices $\{\B{S}_u^{(0)}[k]\}_{u=1,k=1}^{U,K}$ in Alg. 1 to the scaled identity matrix, such that the power constraint holds with equality.

Figures \ref{fig:sumRates_differentMethods} and \ref{fig:sumRates_differentMethods2} show the system performance of the different approaches for two downlink setups. In particular, the setup in Fig. \ref{fig:sumRates_differentMethods} uses $M=16$, $N_\text{RF}=4$, and $K=16$ to enable comparisons with
the high-complexity ITA solution. We have included the achievable sum-rates obtained with both the rank-constrained digital precoders provided by the \ac{RCCA} algorithm (Digital \ac{RCCA}) and after their corresponding factorization to obtain the hybrid version (Hybrid \ac{RCCA}). We consider a relatively large number of channel paths, $N_p=\{16,32\}$, as it is more realistic in a practical scenarios, and makes the precoder design more difficult because of the larger differences between the channel responses corresponding to distant subcarriers. Moreover, the multiuser interference also increases with this parameter. As observed, the proposed \ac{RCCA} algorithm outperforms the other  approaches, especially for medium and high \ac{SNR} regimes. The superior performance of \ac{RCCA} with respect to Hybrid \ac{LISA}  increases for the largest \ac{SNR} values in both setups. This is because \ac{RCCA} considers the interference from all the other users at each algorithm step, whereas LISA only selects one candidate user at each step and removes the interference caused by that user to the previously selected ones. \ac{LISA} is suitable for low scattering scenarios where user subcarriers experience similar channel responses but quickly consumes the available spatial degrees of freedom. Hybrid wideband \ac{ITA} achieves an intermediate performance at the expense of higher complexity. On the other hand, the AM-based and the hybrid precoding procedure of \cite{sohrabi2017hybrid}  are competitive at the low \ac{SNR} regime but they clearly degrade as the \ac{SNR} increases. For the AM-based alternative, this is due to the use of the \ac{MRT} strategy that ignores interference, and the hindrance of imposing a per-subcarrier power constraint. In the case of the method of \cite{sohrabi2017hybrid}, the authors again rely on an analog precoding design that neglects interference, thus limiting its usability in mid and high \ac{SNR} scenarios. As observed, the performance degradation of this method for high SNR values is particularly remarkable as the number of transmit antennas and subcarriers grows.

\begin{figure}[t!]
	\centering
	\includegraphics[width=0.99\linewidth]{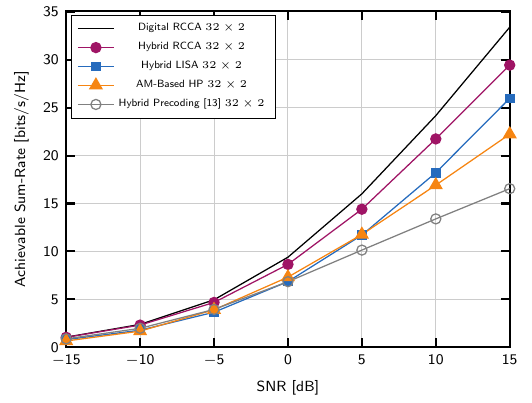}
	\caption{Achievable sum-rates of the different algorithms for a downlink scenario with the following configuration: $M=32, R=2, U=6, N_{\text{RF}}=4, N_p = 32$ and $K=64$.}
	\label{fig:sumRates_differentMethods2}
\end{figure}

 Comparing Digital \ac{RCCA} and Hybrid \ac{RCCA}, we observe that the performance loss caused by the factorization is relatively small. This confirms that the assumption in \eqref{eq:rankApprox} holds in practice, i.e., the use of a common structure to limit the rank of the solutions leads to covariance matrices with approximately the same rank in the downlink. Larger losses for higher \ac{SNR} values are related to the limitations of the factorization algorithms to keep the structure of the digital counterpart \cite{rusu2016low}.

A second computer experiment was carried out to evaluate the impact on performance of the number of channel paths $N_p$. Figure \ref{fig:comparison_bypaths} shows the achievable sum-rates for \ac{RCCA}, Hybrid LISA and AM-based for the setup in Table \ref{tab:Sim1} and different number of paths $N_p \in \{4,8,16,24,32\}$. In addition, two different \ac{SNR} values are considered, namely $-5$ dB and $10$ dB. 
As observed, \ac{RCCA} provides the highest achievable sum-rate even for those situations more favorable to Hybrid \ac{LISA}, i.e., those with smaller number of channel paths $N_p$. In addition, the proposed algorithm efficiently exploits the similarity between the subcarriers even for large values of $N_p$. Note, however, that the performance of all approaches decreases as $N_p$ increases due to the different spatial features among channels for distant subcarriers, and the stronger multiuser interference.

\begin{figure}[t!]
	\centering
	\includegraphics[width=0.99\linewidth]{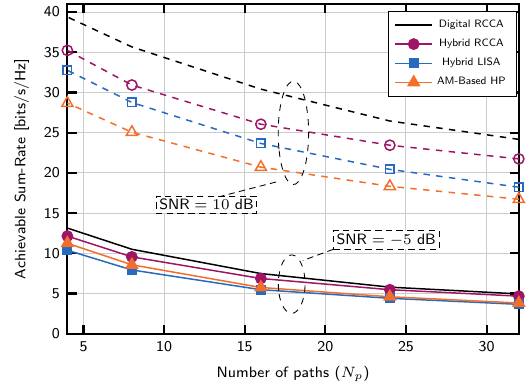}
	\caption{Achievable sum-rates vs. number of reflection paths for the different approaches and two \ac{SNR} values: -5 dB and 10 dB. The configuration setup is: $M=32, R=2, U=6, N_{\text{RF}}=4$ and $K=64$.}
	\label{fig:comparison_bypaths}
\end{figure}

In the next experiment, we are interested in evaluating the behavior of the proposed algorithm depending on the number of RF chains available at the \ac{BS}.  Figure \ref{fig:comparison_byRF} shows the achievable sum-rates of the different approaches considered for the hybrid design of the downlink precoders for a challenging scenario with $U=12$ users and $N_p=32$ paths. We have decided to increase the number of users with respect to Table \ref{tab:Sim1} to better illustrate the dependency on the number of \ac{RF} chains. Like in the previous experiment, we focus on two representative \ac{SNR} values: -5 dB and 10 dB. As observed, the obtained results show the superior performance of the proposed design strategy \ac{RCCA} regardless of the number of \ac{RF} chains. The advantages of using the \ac{RCCA}-based approach are remarkable for the scenario with \ac{SNR} = 10 dB, i.e., high \ac{SNR} regime, and a large number of RF chains. In this case, RCCA provides a gain, in terms of achievable sum-rate, of about 8 bits/s/Hz with respect to Hybrid LISA. Again, this positive behavior of the \ac{RCCA} algorithm is related to a more exhaustive management of the user interference, which is more apparent at the high \ac{SNR} regime. On the other hand, AM-based strategy provides an acceptable performance in the scenario with the lowest \ac{SNR} value but its performance significantly degrades as the \ac{SNR} increases for any number of \ac{RF} chains.

\begin{figure}[t!]
	\centering	\includegraphics[width=0.99\linewidth]{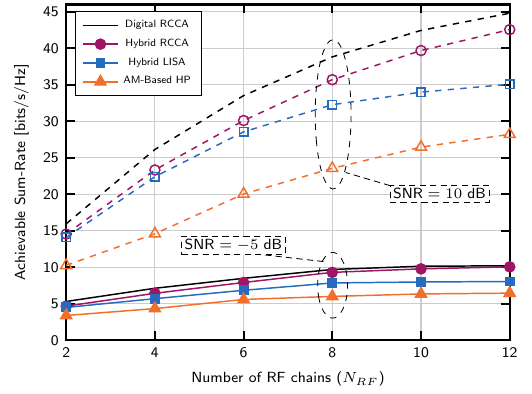}
	\caption{Achievable sum-rates vs. number of RF chains at the BS for the different approaches and two \ac{SNR} values: -5 dB and 10 dB. The configuration setup is: $M=32, R=2, U=12, K=64$ and $N_p = 32$.}
	\label{fig:comparison_byRF}
\end{figure}

\subsection{Performance under practical impairments}
In this section, we evaluate the performance impact caused by two effects that are present in practical setups. More specifically, we first consider an analog hybrid precoding with quantized phases and, second, the more general scenario where the array response vectors are affected by beam squint.

In Fig. \ref{fig:PSresolution} we evaluate the impact of reducing the resolution of the phase shifters employed to build the analog network of the hybrid architecture. Under this assumption, the entries of the analog precoding matrix $\B{P}_\text{RF}$ can only take values from a quantized phases set. In particular, assuming $b$ quantization bits, the available phases are 
$2^b$. According to the obtained results, we can conclude that the proposed method is robust against phase quantization, as the performance loss is moderate even for $b=2$ quantization bits (about 1.5 bits/s/Hz).

\begin{figure}[t!]
	\centering
	\includegraphics[width=0.99\linewidth]{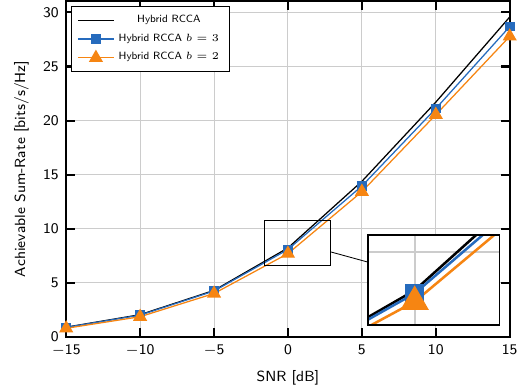}
	\caption{Achievable sum-rates of RCCA for a downlink scenario with the following configuration: $M=32, R=2, U=6, N_{\text{RF}}=4, N_p=32$ and $K=64$ and different number of resolution bits $b={\infty,3,2}$.}
	\label{fig:PSresolution}
\end{figure} 

In the next experiment, we employ the configuration setup of Table \ref{tab:Sim1}  adapting the parameters to a mmWave scenario where  the array response vectors of \eqref{eq:steeringVector} become frequency dependent and the beam squint effect is more remarkable \cite{Wang2019}. In particular, we consider a central frequency $f_c=28$ GHz and different signal bandwidths $B=\{400,800,1600,3200,4000\}$ MHz. From the results observed in Fig. \ref{fig:BS}, we can conclude that the influence of beam squint is  significant for a \ac{SNR} value of $10$ dB and moderate when the \ac{SNR} is $-5$ dB. This is motivated because the spatial similarity of the channel subspaces for different frequencies is compromised by beam squint. Accordingly, the performance losses due to hybrid factorization increase with the bandwidth $B$, as the hybrid approach is less effective when handling multiuser interference.
\begin{figure}[t!]
	\centering
	\includegraphics[width=0.99\linewidth]{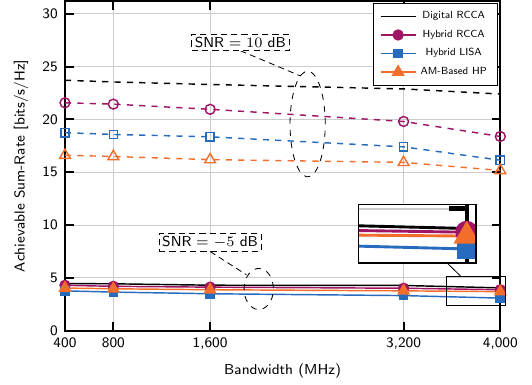}
	\caption{Achievable sum-rates vs. bandwidth for the different approaches and two \ac{SNR} values: -5 dB and 10 dB. The configuration setup is: $M=32, R=2, U=12, N_{\text{RF}}=4, K=64$ and $N_p = 32$.}
	\label{fig:BS}
\end{figure} 

\subsection{Complexity and convergence analysis}

In this subsection, we first compare the computational complexity of the proposed \ac{RCCA} algorithm with respect to the competitors. Table \ref{tablecomplejidad} shows the complexity order of these five wideband hybrid precoding approaches: \textit{I)} The complexity of \ac{RCCA} is mainly determined by the step 6 in Algorithm \ref{alg:iterWfillRank}, where the calculation of the SVD of the matrix which results from stacking the equivalent channels corresponding to all the subcarriers for the $u$-th user is required. The complexity of this step is in the order of $\mathcal{O}(KMR^2)$ for the considered setups. \textit{II)} The most computational expensive operation in Hybrid LISA is the calculation of the precoding candidates to select the best user at each algorithm step \cite{gonzalez2019hybrid}. The complexity order of this operation is $\mathcal{O}\left(\operatorname{min}\{KM^2,K^2M\}\right)$ since all the frequency equivalent channels for each user are considered. \textit{III)} The complexity order of the AM-based method was determined in \cite{Yuan19}, and it comes from the computation of the matrix  product of the digital precoders for all the subcarriers and users times the analog precoding matrix. As observed, the complexity order of the AM-based method  is actually lower than that of the other  approaches since it is linear on the number of transmit antennas and subcarriers. \textit{IV)} The hybrid precoding design of \cite{sohrabi2017hybrid} also presents a computational efficient solution, as the complexity order $\mathcal{O}(M^2)$ comes from the computation of the average covariance matrices for the wideband channels. These products can be obtained in parallel for all the subcarriers. \textit{V)} The wideband extension of ITA has a complexity order significantly higher than the other approaches. This is because the general-purpose solvers for semidefinite programming (SDP) problems are based on interior point methods whose complexity is in the order of $\mathcal{O}(mn^2 + n^3)$ for a SDP problem with $m$ linear matrix inequalities and $n$ variables \cite{sdpt3}.

\begin{table}[t!]
	\centering
	\caption{{Computational complexity of wideband hybrid precoders}.}
	\setlength{\tabcolsep}{5pt}
	\def\arraystretch{1.3}
	\label{tablecomplejidad}
	\begin{tabular}{ll}
		\hline
		\textbf{Algorithm}&  \textbf{Complexity order}\\ \hline
		\ac{RCCA} & {$\mathcal{O}\left(\operatorname{max}\{KMR^2,M^2R\}\right)$} \\ \hline
		Hybrid LISA& $\mathcal{O}\left(\operatorname{min}\{KM^2,K^2M\}\right)$ \\ \hline
		AM-Based HP  & $\mathcal{O} \left( KMUN_{{\text RF}}\right)$ \\ \hline
		Wideband ITA& $\mathcal{O}~(K^3 M^6)$ \\ \hline
		Hybrid Precoding in \cite{sohrabi2017hybrid}& $\mathcal{O}~( M^2)$ \\ \hline
	\end{tabular}
\end{table}

In order to illustrate the influence of the number of antennas $M$ on the overall computational costs of the the hybrid precoder designs considered in this section, we provide in Fig. \ref{fig:complexity} an approximation to the number of operations required for each method in a logarithmic scale. The result shows that Hybrid LISA and the posed RCCA method increase their computational costs faster than other computationally cheap solutions.  In addition, Hybrid Wideband ITA becomes hardly applicable even for moderate numbers of antennas. 

\begin{figure}[!t]
	\centering
	\includegraphics[width=0.99\linewidth]{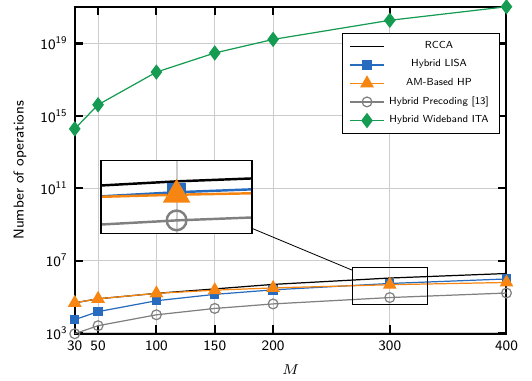}
	\caption{Computational complexity vs. number of antennas at the BS $M$ for the different approaches. The configuration setup is: $R=2$, $U=6$, $N_{\text{RF}}=4$ and $K=64$.}
	\label{fig:complexity}
\end{figure}

Now, we empirically analyze the convergence of the \ac{RCCA} algorithm.
As commented, the particular structure of the user covariance matrices and the non-convex nature of the problem imposed by the rank restrictions make it impossible to theoretically guarantee the convergence of the \ac{RCCA} algorithm. However, the obtained simulation results and the empirical convergence analysis carried out during the computer experiments suggests that the \ac{RCCA} algorithm presents good convergence properties. This is in part motivated by the use of numerical techniques which ensure the convergence of the algorithm to a local optimum, following a similar approach to \cite{JiReViGo05}. An example of this is observed in Figure \ref{fig:convergence} which shows the achievable sum-rates obtained with the \ac{RCCA} algorithm at the end of each iteration for three different \ac{SNR} values, namely -5 dB, 5 dB and 15 dB. As observed, the algorithm rapidly converges after a few iterations in the three \ac{SNR} cases. Usually, the number of iterations required to achieve convergence increases with the number of channel paths $N_p$, the number of users $U$, and the SNR; that is, for scenarios where interference management is difficult. As already mentioned, due to the complicated structure of the search space of the original problem formulation in (4), the optimum achievable sum-rate remains unknown. Thus, it is impossible to check if convergence to the global optimum has been reached.

\begin{figure}[!t]
	\centering	\includegraphics[width=0.99\linewidth]{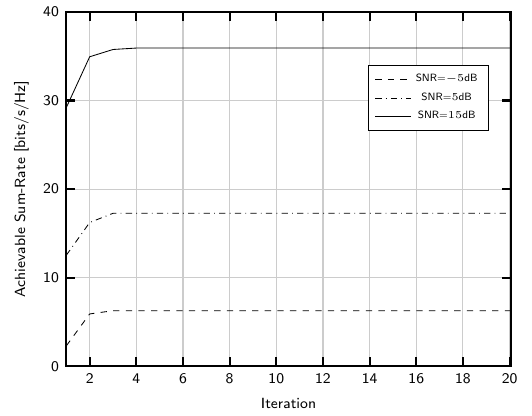}
	\caption{Achievable sum-rates vs. number of iterations for the \ac{RCCA} algorithm and \ac{SNR} values: -5 dB, 10 dB, and 15 dB. The configuration setup is: $M=32, R=2, U=6, N_{\text{RF}}=6$ and $K=64$.}
	\label{fig:convergence}
\end{figure}

\vspace*{-0.3cm}
\section{Conclusions}

An innovative approach to hybrid precoding at the BS of a downlink massive \ac{MIMO} wideband system, termed \ac{RCCA}, has been developed. The approach aims at exploiting the spatial similarity among the different \ac{OFDM} subcarriers by forcing  digital solutions with a low rank. In addition, the inter-user interference is considered to approximate equivalent channel gains which in turn are used to update the transmit covariance matrices.  The numerical experiments show that the proposed \ac{RCCA} approach outperforms both direct extensions of previous rank-constrained designs and other benchmark solutions such as hybrid \ac{LISA} or AM-based algorithms. In this sense, \ac{RCCA} has been shown to be a robust solution for a wide range of practical \ac{SNR} values and scenarios with different number of reflection channel paths.

\section*{Acknowledgments}
This work has been supported in part by grants ED431C 2020/15 and ED431G 2019/01 (to support the Centro de Investigación de Galicia “CITIC”) funded by Xunta de Galicia and ERDF Galicia 2014-2020, and by grants PID2019-104958RB-C42 (ADELE) and PID2020-118139RB-I00 funded by MCIN/AEI/10.13039/501100011033. The authors thank the Defense University Center at the Spanish Naval Academy (CUD-ENM) for all the support provided for this research.
 
\appendices
\section{Frequency-flat covariance matrix structure} \label{app:covariance_structure}
We focus on the case where the covariance matrix has to be designed for a single subcarrier, $K=1$. Then, the KKT condition in \eqref{eq:KKTcondition} is equivalent to
\begin{equation}
	\left(\mathbf{I}+\B{H}_{\bar{u}}\B{H}_{\bar{u}}^\He\B{S}_u\right)^{-1}\B{H}_{\bar{u}}\B{H}_{\bar{u}}^\He-\lambda_u\mathbf{I}=\mathbf{0}.
	\label{eq:singleCarrierCond}
\end{equation}
As observed, if the structure of the covariance matrix $\B{S}_u$ forces the first term \eqref{eq:singleCarrierCond} to be positive definite, then the subspace spanned by the columns of $\B{S}_u$ achieves the condition in \eqref{eq:singleCarrierCond} with equality for certain spectrum of $\B{S}_u$. To show this statement we introduce the decomposition $\B{H}_{\bar{u}}\B{H}_{\bar{u}}^\He = \B{U}\B{\Gamma}\B{U}^\He$, with unitary matrix $\B{U}$ and $\B{\Gamma}=\diag(\gamma_1,\ldots,\gamma_R)$ with $\gamma_i\geq 0$, $\forall i\in\{1,\ldots,R\}$. We consider that the elements $\gamma_1$ to $\gamma_j$  are strictly lager than $0$ without loss of generality. Now, we rewrite the left-hand side of \eqref{eq:singleCarrierCond} as
\begin{align}
	&\left(\mathbf{I}+\B{U}\B{\Gamma}\B{U}^\He\B{S}_u\right)^{-1}\B{U}\B{\Gamma}\B{U}^\He+\B{U}\B{L}\B{U}^\He-\B{U}\B{L}\B{U}^\He-\lambda_u\mathbf{I}\notag\\
	&=\left(\B{U}(\B{\Gamma}+\B{L})^{-1}\B{U}^\He+\B{U}(\B{\Gamma}+\B{L})^{-1}\B{\Gamma}\B{U}^\He\B{S}_u\right)^{-1}\notag\\
	&\quad-\B{U}\B{L}\B{U}^\He-\lambda_u\mathbf{I},
\end{align} 
where we have also introduced the auxiliary diagonal matrix $\B{L}=\diag(0,\ldots,0,l_{j+1},\ldots,l_R)$. Next, we select the following structure for the covariance matrix $\B{S}_u=\B{U}\B{\Psi}\B{U}^\He$ to get 
\begin{align}
&\left(\B{U}\left[(\B{\Gamma}+\B{L})^{-1}+\tilde{\mathbf{I}}\B{\Psi}\right]\B{U}^\He\right)^{-1}-\B{U}(\B{L}+\lambda_u\mathbf{I})\B{U}^\He\notag\\
&=\B{U}\left(\left[(\B{\Gamma}+\B{L})^{-1}+\tilde{\mathbf{I}}\B{\Psi}\right]^{-1}-(\B{L}+\lambda_u\mathbf{I})\right)\B{U}^\He,
\label{eq:aux_step}
\end{align} 
with $\tilde{\mathbf{I}}=\diag(1,\ldots,1,0,\ldots,0)$ having $j$ ones in the diagonal, followed by $R-j$ zeros. In addition, note that $\B{\Psi}=\diag(\psi_1,\ldots,\psi_R)$. The expression in \eqref{eq:aux_step} will be equal to the zero matrix if the two terms inside the brackets cancel each other out, that is, if
\begin{equation*}
	(\B{\Gamma}+\B{L})^{-1}+\tilde{\mathbf{I}}\B{\Psi}=(\B{L}+\lambda_u\mathbf{I})^{-1}. %
\end{equation*}
We now distinguish between the first $j$ elements and the remaining ones, leading to the following conditions
\begin{equation}
	\begin{cases}
		\frac{1}{\gamma_i}+\psi_i=\frac{1}{\lambda_u}\quad&\text{if}\;i\leq j\\
		\frac{1}{l_i}=\frac{1}{l_i+\lambda_u}\quad&\text{otherwise}.
	\end{cases}
\end{equation}
While the first condition can be achieved by setting $\psi_i=\frac{1}{\lambda_u}-\frac{1}{\gamma_i}$, the second condition approximately holds when we choose $l_i\gg \lambda_u$, $\forall i$.  We hence conclude that, selecting the structure of the user covariance matrix as $\B{S}_u=\B{U}\B{\Psi}\B{U}^\He$, the first term in \eqref{eq:singleCarrierCond} results in a positive definite matrix, which is meaningful in a per-carrier basis.

\section{Analysis of the optimality of $N_s=N_\text{RF}$}
\label{app:streams}
In this appendix, we evaluate the effect of including an additional stream when $N_s=N_\text{RF}$ in \eqref{eq:rankFormulation} holds to show that, even under the assumption of perfect decoding, this will lead to a reduction in the achievable sum-rate. %

We consider that the matrices $\B{S}_u[k]$ allocate streams to a certain set of users $\mathcal{U}$ such that the rank of the overall transmit covariance is $N_{\text{RF}}$, while providing the largest equivalent channel gains, according to Sec. \ref{eq:power_alloc}. The achievable sum-rate obtained under the power constraint $\sum_{k=1}^K\sum_{u=1}^U\trace(\B{S}_u[k])\leq P_\text{tx}$ is denoted by $R$, and for simplicity we assume $\sigma_n^2=1$. In this context, if a new stream is allocated to user $j\notin\mathcal{U}$, the achievable sum-rate including the new stream reads as
\begin{equation}
R^\prime=\sum_{u\in\mathcal{U}}R_u^\prime+R_j^\prime,
\end{equation}
where $R^\prime_u$ is the achievable rate for user $u$ when considering the appropriate set of transmit covariance matrices $\B{S}_u^\prime[k]$, and $R^\prime_j$ is the achievable rate for the $j$-th user which is given by
\begin{equation}
R_j^\prime=\sum_{k=1}^K\log_2\det\left(\mathbf{I}+\B{H}_{\bar{j}}^\He[k]\B{S}^\prime_j[k]\B{H}_{\bar{j}}[k]\right),
\end{equation}
with $\B{H}_{\bar{j}}[k]$ being the effective channels. As we keep the rank of the overall transmit covariance matrix constant in the uplink, the $j$-th user covariance matrix is given by
\begin{equation*}
\B{S}_j^\prime[k]=\sum_{u\in\mathcal{U} }\B{V}_u\B{P}_{j,u}[k]\B{V}_u^\He,
\end{equation*}
where $\B{S}_u^\prime[k]=\B{V}_u\B{P}_{u}[k]\B{V}_u^\He$ has the structure in \eqref{eq:S} for all users and subcarriers, with semi-unitary $\B{V}_u\in\mathbb{C}^{M\times N_{s,u}}$, and $\B{P}_{j,u}[k]\succeq\mathbf{0}$ are the diagonal matrices that determine the power allocation for the $j$-th user.

By employing the effective channel  $\B{H}_{\bar{j}}[k]$, the KKT condition derived in App. \ref{app:covariance_structure} can be extended to consider multiple carriers as [cf. \eqref{eq:singleCarrierCond}]
\begin{equation*}
\left(\mathbf{I}+\B{H}_{\bar{j}}[k]\B{H}_{\bar{j}}^\He[k]\B{S}_j^\prime[k]\right)^{-1}\B{H}_{\bar{j}}[k]\B{H}_{\bar{j}}^\He[k]-\lambda_j[k]\mathbf{I}=\mathbf{0}.
\end{equation*}
Note that the equation above cannot be fulfilled in general due to the lack of flexibility in the design of $\B{S}^\prime_j[k]$. The reason behind this claim is that the subspaces spanned by the rows and columns of $\B{S}^\prime_j[k]$ are rotations of those for $\B{V}_u$, and they are not specifically designed for user $j$. Accordingly, the KKT conditions would not be reachable simultaneously for all the users and, therefore, better power distributions would be possible.  This situation remains the same as long as user $j$ is active, i.e., the performance metric can be improved by assigning all the available power to the users in the set $\mathcal{U}$. Therefore, we conclude that $R^\prime\leq R$. 

Moreover, for the sake of completeness, we investigate the impact of the achievable sum-rate for a more practical approach, i.e., when successive cancellation is not considered. In this case, Zero-Forcing (ZF) precoding achieves channel capacity in the high \ac{SNR} regime \cite{YoGo06,GuUtDi09}. Therefore, we consider a set of downlink transmit covariance matrices $\B{Q}_u[k]=\B{V}_u\B{P}_u[k]\B{V}_u^\He$, with semi-unitary $\B{V}_u\in\mathbb{C}^{M\times N_{s,u}}$ and $\B{P}_u[k]\in\mathbb{C}^{N_{s,u}\times N_{s,u}}$, such that the equivalent channel gains for users $u\in\mathcal{U}$ are larger than the gains for the remaining ones $j\notin\mathcal{U}$ for given covariance matrices $\B{Q}_u[k]$, i.e., $\B{H}_u[k]\B{V}_u\B{V}_u^\He\B{H}_u^\He[k]\succeq\B{H}_{\bar{j}}[k]\B{V}_u\B{V}_u^\He\B{H}_{\bar{j}}^\He[k]$, and $\B{H}_i[k]\B{Q}_u[k]\B{H}_i^\He[k]=\mathbf{0},\forall k$, with $u,i\in\mathcal{U}$ and $j\notin\mathcal{U}$. In other words, the set of covariance matrices $\B{Q}_u[k]$, for $u\in\mathcal{U}$, are optimal under the assumption  $N_s=N_\text{RF}$. The achievable sum-rate is then a simplified version of that in \eqref{eq:rankFormulation}, i.e.,
\begin{equation*}
R=\sum_{k=1}^K \sum_{u\in\mathcal{U}}\log_2\det\left(\mathbf{I}+\B{H}_u[k]\B{Q}_u[k]\B{H}_u^\He[k]\right).
\end{equation*} 

When a new stream is allocated to user $j\notin\mathcal{U}$, and with the aim of satisfying the rank constraint, the covariance matrix for user $j$ reads as 
\begin{equation}
\B{Q}_j[k]=\sum_{u\in\mathcal{U} }\B{V}_u\B{P}_{j,u}[k]\B{V}_u^\He.
\label{eq:convMatrRank}
\end{equation}

In this scenario, the zero-forcing condition for user $j$ does not hold, leading to the following achievable rate for the user $u$ and the subcarrier $k$
\begin{align*}
R^\prime_u[k]=\log_2\frac{\det\left(\mathbf{I}+\B{H}_u[k]\big(\B{Q}_u[k]+\B{Q}_j[k]\big)\B{H}_u^\He[k]\right)}
{\det\left(\mathbf{I}+\B{H}_u[k]\B{Q}_j[k]\B{H}_u^\He[k]\right)}.
\end{align*}
Using the definition in \eqref{eq:convMatrRank} and the zero interference assumption, the former expression is rewritten as
\begin{equation*}
\hspace*{-0.1cm}R^\prime_u[k]=\log_2\frac{\det\left(\mathbf{I}+\B{H}_u[k]\B{V}_u\big(\B{P}_u[k]+\B{P}_{j,u}[k]\big)\B{V}_u^\He\B{H}_u^\He[k]\right)}
{\det\left(\mathbf{I}+\B{H}_u[k]\B{V}_u\B{P}_{j,u}[k]\B{V}_u^\He\B{H}_u^\He[k]\right)}.
\end{equation*}
Therefore, since the power available for the users $u\in\mathcal{U}$ is smaller due to the incorporation of the new stream, the achievable sum-rate for $u\in\mathcal{U}$ can be upper bounded as follows
\begin{align}
\sum_{k=1}^K\sum_{u\in\mathcal{U}}R^\prime_u[&k]\leq R \label{eq:boundSumRate}\\- &\sum_{k=1}^K\sum_{u\in\mathcal{U}}\det\left(\mathbf{I}+\B{H}_u[k]\B{V}_u\B{P}_{j,u}[k]\B{V}_u^\He\B{H}_u^\He[k]\right).\notag
\end{align}
For the new user $j$, it is also possible to obtain  upper bounds for $R^\prime_j[k]$ as follows
\begin{align}
R^\prime_j&[k]=\log_2\det\left(\mathbf{I}+\B{H}_{\bar{j}}[k]\sum_{u\in\mathcal{U}}\B{V}_u\B{P}_{j,u}[k]\B{V}_u^\He\B{H}_{\bar{j}}^\He[k]\right)\notag\\
&\leq \log_2\det\left(\mathbf{I}+\sum_{u\in\mathcal{U}}\B{H}_u[k]\B{V}_u\B{P}_{j,u}[k]\B{V}_u^\He\B{H}_u^\He[k]\right),
\label{eq:jUserRateBound}
\end{align}
where we have used that the covariance matrices are designed for users $u\in\mathcal{U}$. Thus the equivalent channel gains are larger than for users not included in the set $\mathcal{U}$. Combining the bounds in \eqref{eq:boundSumRate} and \eqref{eq:jUserRateBound}, and introducing the auxiliary positive semi-definite matrices $\B{A}_u[k]=\B{H}_u[k]\B{V}_u\B{P}_{j,u}[k]\B{V}_u^\He\B{H}_u^\He[k]$, we obtain 
\begin{align*}
\sum_{k=1}^K\sum_{u\in\mathcal{U}}R^\prime_u+R^\prime_j\leq R+\sum_{k=1}^K\log_2\frac{\det\left(\mathbf{I}+\sum_{u\in\mathcal{U}}\B{A}_u[k]\right)}{\prod_{u\in\mathcal{U}}\det\left(\mathbf{I}+\B{A}_u[k]\right)}.
\end{align*}
Now, the achievable sum-rate not increasing after allocating a new stream is equivalent to proving that 
\begin{equation}
\det\big(\mathbf{I}+\sum_{u=1}^{|\mathcal{U}|}\B{A}_u[k]\big)\leq \prod_{u=1}^{|\mathcal{U}|}\det\left(\mathbf{I}+\B{A}_u[k]\right).
\label{eq:inequalityPerCarrier}
\end{equation}
Rewriting the left-hand side of this inequality, we get
\begin{align*}
&\prod_{u=1}^{|\mathcal{U}|}\det\left(\mathbf{I}+\big(\mathbf{I}+\sum_{i=1}^{u-1}\B{A}_i[k]\big)^{-1}\B{A}_u[k]\right)\leq\\ &\hspace{5.5cm}\prod_{u=1}^{|\mathcal{U}|}\det\left(\mathbf{I}+\B{A}_u[k]\right).
\end{align*}
This last inequality is proven by showing that, for positive semidefinite matrices $\B{A}$ and $\B{B}$, the following inequality holds
\begin{equation}
\det\left(\mathbf{I}+\big(\mathbf{I}+\B{B}\big)^{-1}\B{A}\right)\leq \det\left(\mathbf{I}+\B{A}\right).
\label{eq:detIneq}
\end{equation}
Denoting as $\lambda_n(\B{X})$ the $n$-th largest eigenvalue of $\B{X}$, we have that \cite[p. 423]{matrixAnalysis}
\begin{equation}
\lambda_n(\big(\mathbf{I}+\B{B}\big)^{-1}\B{A})\leq \lambda_1\left(\big(\mathbf{I}+\B{B}\big)^{-1}\right)\lambda_n(\B{A}). 
\end{equation} 
We therefore conclude that the eigenvalues of $\big(\mathbf{I}+\B{B}\big)^{-1}\B{A}$ are smaller than or equal to those of $\B{A}$. Correspondingly, the inequalities in  \eqref{eq:detIneq} and \eqref{eq:inequalityPerCarrier} hold, thus completing the proof.

\bibliographystyle{IEEEtran}

\bibliography{refs}

\end{document}